\documentclass[a4paper]{article}
\usepackage[english]{babel}
\usepackage[applemac]{inputenc}
\usepackage[T1]{fontenc}
\usepackage{bbm,mathtools}
\usepackage{amsmath,color}
\usepackage{amssymb}
\usepackage{amsfonts,bbm}
\usepackage{amsthm,slashed}
\usepackage[all,cmtip]{xy}

\def\vN{\mathbb{N}}
\def\vZ{\mathbb{Z}}

\def\vR{\mathbb{R}}

\def\vF{\mathbb{F}}

\newcommand{\Z}{\mathbbm{Z}}
\newcommand{\qsp}[2]{\,\ensuremath{\raise.5ex\hbox{$#1$}\big\slash\raise-.5ex\hbox{$#2$}}} 

\newtheorem{teo}{Theorem}
\newtheorem{problem}{Problem}
\newtheorem{prop}[teo]{Proposition}
\newtheorem{lemma}[teo]{Lemma}
\newtheorem{definition}{Definition}

\theoremstyle{definition}
\newtheorem{remark}{Remark}

\title{A general construction for monoid-based knapsack protocols}
\author{G. Micheli, M. Schiavina \\ Institut f\"ur Mathematik, \\Universit\"at Z\"urich}
\date{}

\begin{document}
\maketitle

\begin{abstract}
We present a generalized version of the knapsack protocol proposed by D. Naccache and J. Stern at the Proceedings of Eurocrypt (1997). Our new framework will allow the construction of other knapsack protocols having similar security features.
We will outline a very concrete example of a new protocol using extension fields of a finite field of small characteristic instead of the prime field $\vZ/p\vZ$,  but more efficient in terms of computational costs for { asymptotically equal information rate and similar key size}.
\end{abstract}

\section{Introduction}
Building new asymmetric encryption schemes has always been one of the main goals of cryptographers. After the idea of public key cryptography was presented in \cite{DH}, only few more public key encryption schemes were developed such as the  RSA  \cite{RSA}, the El Gamal \cite{ ELG}, the McEliece cryptosystem \cite{MCE}, the NTRU \cite{NTRU} or the HFE \cite{HFE} (for an overview \cite{OV,SAL}). Some new ideas for building new cryptographic schemes based on semigroup actions can also be found in \cite{ROS}, while in the context of knapsack quantum cryptographic schemes we refer for instance to \cite{QC}. {What D. Naccache and J. Stern built in \cite{NSK} was a proposal for an asymmetric protocol (NSK) following the earlier ideas of Morii and Kasahara in \cite{MK1}, further developed by Kasahara et al. in \cite{MK2,MK3}. The NSK protocol consists of a shuffling modulo $p$ of an easy problem over the integers, i.e. the factorization of a composite integer where the prime factors are chosen among a fixed set of small size.}
Given $p$ a prime and $\vZ/p\vZ$ the finite field of remainder classes, the NSK protocol is based on the unique factorization property of $\vZ$, which guarantees the uniqueness of the encryption.

This approach can be generalized to the case of multiplicative monoids (Section \ref{sec1}), and the NSK protocol is just a particular instance for the monoid $(\vZ,\cdot)$ of the general framework (subsection \ref{particularcase}).
Using this new general setting we are able to construct an analogous of the NSK protocol relying on the unique factorization properties of $\vF_q[x]$, instead of $\vZ$, where $\vF_q$ is the finite field of order $q$ (Section \ref{poly}).
The security of our particular proposal will rely on the arithmetic structure of the finite field $\vF_q[x]/(h(x))$ for some $h(x)\in \vF_q[x]$, irreducible of suitable degree (instead of the finite field of remainder classes $\vZ/p\vZ$). One of the main advantages of this kind of setting is that the security is based on an exponentiation over a finite field in such a way that it will be unfeasible for an attacker even to set up a discrete logarithm problem (DLP). Indeed, as we will show in the following, since the optimal version of the NSK protocol requires that the chosen prime be next to $\prod_i p_i$, the factorization of $p-j$ for some small $j$ could allow for a reduction to a DLP.
In our case, instead, we choose a set of irreducible polynomials and fix the degree of the reducing polynomial. By doing so there is no information leakage. 
Our new structural conditions will be related only to the degree of the carrier polynomials used for the encryption, avoiding any kind of DLP reduction. 

In subsection \ref{security} some issues concerning the security of the protocol will be addressed, in particular to avoid \emph{subgroup attacks}, that could possibly lead to information.

This new setting will lead to some advantages in terms of computational costs of encryption and decryption. In fact, arithmetics over finite fields $\vF_{q^m}$ is considered to be preferrable than arithmetics over $\vZ_{p}$ when $p\simeq q^m$ and $q\ll p$ in terms of computations. We will analyse the key features of our protocol, such as the number of parameters involved for the setting up of the public key, and this will allow us to show a greater deal of flexibility, in comparison with the NSK protocol.

{In subsection \ref{compare} we will analise the asymptotics of the information rate of our protocol, showing that it is equal to that of \cite{NSK}. An exact formula for the information rate will also be provided.}

As a subproduct, we present in Section \ref{CRT} a variation of the polynomial protocol where the irreducibility of $h(x)$ is dropped. The encryption is performed over a suitable direct sum of fields, and a decryption is available thanks to the Chinese Remainder Theorem.

\section{The new class}\label{sec1}
In this section we will present a generalized version of the protocol presented in \cite{NSK}. 

Let $S$ be a monoid and $\sim$ a finite index congruence on $S$. We will denote the class of an element $s\in S$ with respect to $\sim$ as $[s]$.

\begin{definition}
A morphism $\psi$ will be said to be $\sim$proper, if 
\begin{itemize}
\item $\psi\colon S\longrightarrow S$ is injective
\item $\psi$ is compatible with $\sim$  (i.e. $\psi(x)\sim\psi(y)$ iff $x\sim y$)
\item the induced application $\widetilde{\psi}\colon {S}/{\sim}\longrightarrow {S}/{\sim}$ is invertible.
\end{itemize}
\end{definition}

\begin{definition}\label{defcryptable}
Given $L\in\mathbb{N}$ we will say that $S$ is $L$-\emph{cryptable} under $\sim$ if there exists a $\sim$proper morphism $\psi$ and elements $s_1,\dots,s_L\in S$ such that
\[\alpha_\sim^\psi: \vZ_2^L \longrightarrow S/\sim\]
\[m=(m_1,\dots,m_L)\mapsto \left[\prod^L_{i=1}\psi(s_i)^{m_i}\right]\]
is an injective application.
\end{definition}

The following proposition will be useful later on
\begin{prop}\label{injalpha} Given a monoid $S$ that is $L$-cryptable under $\sim$, the following maps are also injective:
\[\alpha^\psi: \vZ_2^L \longrightarrow S\]
\[(m_1,\dots,m_L)\mapsto \prod^L_{i=1}\psi(s_i)^{m_i}\]
\[\alpha_\sim: \vZ_2^L \longrightarrow S/\sim\]
\[(m_1,\dots,m_L)\mapsto \left[\prod^L_{i=1}s_i^{m_i}\right]\]
\[\alpha: \vZ_2^L \longrightarrow S\]
\[(m_1,\dots,m_L)\mapsto \prod^L_{i=1}s_i^{m_i}.\]
\end{prop}
\begin{proof}
The proof follows by observing that, since $\psi$ is $\sim$proper morphism, then also $\alpha_\sim$ is injective. Also $\alpha_\sim^\psi$ injective implies that $\alpha^\psi$ is injective. 
Again, since $\psi$ is an injection, also $\alpha$ is injective.
\end{proof}

{As we have already pointed out}, this properties are necessary to keep the encryption meaningful. In the following we will see how it is possible to find non trivial examples of this construction.

Now, denote the image of  any map $f$ between sets by $\Im(f)$, and consider the following problems:
\begin{problem}\label{problem1}
Given $c\in \Im(\alpha^\psi_\sim)$ find $m$ such that $\alpha^\psi_\sim(m)=c$
\end{problem}

\begin{problem}\label{problem2}
Given $c'\in \Im(\alpha_\sim)$ find $m$ such that $\alpha_\sim(m)=c'$.
\end{problem}

Let now $S$, be an $L$-cryptable monoid under a congruence $\sim$. Whenever a given triple $(S,\sim,\psi)$ is such that {Problem \ref{problem1}} is difficult, {Problem \ref{problem2}} is easy we define a cryptosystem as follows. Let
 \[(S,\sim,L,\widetilde{\psi}([s_1]),\dots,\widetilde\psi([s_L]))\] 
be the public key and
 \[(\widetilde{\psi}^{-1}, s_1,\dots,s_L)\] 
be the secret key, the main operations are given by
\begin{itemize}
\item{\emph{Encryption}: $E(m):=\alpha^\psi_\sim(m) = \prod^L_{i=1}\widetilde{\psi}([s_i])^{m_i}=: c$}
\item{\emph{Decryption}: $D(c)$ is given by solving Problem \ref{problem2} for $c'=\widetilde{\psi}^{-1}(c)$.}
\end{itemize}
\begin{remark}
The reader should observe that in the definition of the protocol we did not use the injectivity of $\psi$ nor the fact that $S/\sim$ is a quotient of a monoid $S$. This is nevertheless the case in all the examples of this protocol we could find, where Problem \ref{problem2} is easy since a \emph{suitable} lift to $S$ is given. Indeed, in practical situations the problem will be solved computing $(\alpha^{-1}\circ\Gamma)(c')$ where $\Gamma$ is a lift $S/\sim\longrightarrow S$ such that the following diagram 

\begin{align}\label{liftdiagramma}
\xymatrix{
 \Z_2^L  \ar[r]^\alpha \ar[rd]_{\alpha_\sim}  & \Im(\alpha)  \\
  & \Im (\alpha_\sim) \ar[u]_{\widehat{\Gamma}} }
\end{align}
commutes when $\widehat{\Gamma}:=\Gamma |_{\Im(\alpha_\sim)}$

\end{remark}
\begin{remark}\label{definforate}
Notice that the information rate is given by $L/b$ where $b$ is the number of bits that are needed to represent an element of $S/\sim$
\end{remark}

In what follows we will show how the NSK protocol fits in this rather general framework, as well as brand new protocols involving polynomials over finite fields.

\subsection{NSK as a particular instance}\label{particularcase}
In this section we will show how the Naccache-Stern (NSK) protocol fits in our general framework, in the case $S=(\mathbb{Z},\cdot)$.

Consider the prime ideal $P=\left<p\right>$ generated by a prime number $p\in\mathbb{Z}$. Let us denote by $\sim$ the congruence induced by the ideal $P$. Such a congruence is obviously of finite index. Let $v$ be a positive integer with $u=v^{-1}\mod{p-1}$, and let
\[
\begin{aligned} \psi\colon\mathbb{Z}&\longrightarrow \mathbb{Z}\\ a&\longmapsto a^v.\end{aligned}
\]
It can be easily checked that $\psi$ is a $\sim$proper morphism of $\vZ$.

Now choose $L$ distinct prime numbers $p_i$ such that $\prod_{i=1}^L p_i<p$. 

\begin{prop}The map
\begin{equation}\begin{aligned}
\alpha_\sim^\psi\colon \vZ_2^L&\longrightarrow\vZ/p\vZ\\
(m_1,\dots,m_L)&\longmapsto\left[\prod\limits_{i=1}^Lp_i^{m_iv}\right]
\end{aligned}\end{equation}
is an injection and $(\vZ,\cdot)$ is therefore $L$-cryptable under the relation induced by the ideal generated by $p$.
\end{prop} 
\begin{proof}
Assume that there exist two $L$-tuples $(m_1,\dots,m_L),(n_1,\dots,n_L)$ such that $\alpha_\sim^\psi(m_1,\dots,m_L)=\alpha_\sim^\psi(n_1,\dots,n_L)$, then:
\[
\left[\prod\limits_{i=1}^Lp_i^{m_iv}\right]=\left[\prod\limits_{i=1}^Lp_i^{n_iv}\right]\Rightarrow \left[\prod\limits_{i=1}^Lp_i^{m_iv}\right]^u=
\left[\prod\limits_{i=1}^Lp_i^{n_iv}\right]^u\Leftrightarrow
\left[\prod\limits_{i=1}^Lp_i^{m_i}\right]=\left[\prod\limits_{i=1}^Lp_i^{n_i}\right]
\]
in $\vZ/p\vZ$. Since $\prod_{i=1}^Lp_i^{m_i}$ and $\prod_{i=1}^Lp_i^{n_i}$ are smaller than $p$ we also have 
\begin{equation}\label{canlift1}
\prod\limits_{i=1}^Lp_i^{m_i}=\prod\limits_{i=1}^Lp_i^{n_i}
\end{equation} 
in the unique factorization domain $\vZ$, which implies $m_i=n_i\ \forall i$.
\end{proof}

\begin{remark}
Notice that we are able to express equation \eqref{canlift1} because we can always consider the canonical representative $x\in\{0,\dots,p-1\}$ in the remainder class modulo $p$. This representative is also the only representative in $\Im(\alpha)$ by construction, and therefore we have a canonical lift satisfying \eqref{liftdiagramma}.
\end{remark}

\begin{remark}\label{NSKleak}
The reader should observe that when $p=t+\prod_i p_i $ for $t$ small, than the information rate is maximal. Unfortunately in this case factoring $p-t$ is easy because $p-t$ is $p_L$-smooth and $p_L\ll p $, and this gives informations about the bare carriers $p_i$'s. Indeed in this case breaking the NSK protocol is not harder than solving the DLP for the $p_i$'s. Nevertheless the protocol remains interesting for additional features like \cite[Section 3]{NSK}.
\end{remark}

\section{A polynomial version}\label{poly}
In this section we  give a version of the protocol that works over $\vF_{q^d}$ instead of $\vZ/p\vZ$ in such a way that $q^d$ will be
of the same order of magnitude than the size $p$ of the field $\vZ/p\vZ$ in the NSK but $q\ll p$.
In this case the specific difficult problem we want to rely on is the following
\begin{problem}
Let $\vF$ be a finite field and $L\in\vN$.Given $y_1,\dots, y_L\in \vF$,
\[\alpha: \vZ_2^L \longrightarrow \vF\]
\[\alpha(m)=\prod_i y_i^{m_i}\]
 and $c\in \Im(\alpha)$, find $m$ such that $\alpha(m)=c$.
\end{problem}

Let now $k=\vF_q$ and  $k[x]$ the polynomial ring in one variable over $k$.
Let $h(x)$ be an irreducible element in $k[x]$ of degree $d$.
Set $\sim$ to be the congruence associated to the ideal $H=\left<h(x)\right>$ generated by the irreducible polynomial $h(x)$. An efficient algorithm to find irreducible polynomials of fixed degree is given, for instance in \cite{POLLEG}.
Set \[S=(k[x],\cdot)\] and \[S':=S/\sim\;=\left((k[x]/H)^*,\cdot\right)\]
where $(k[x]/H)^*=(k[x]/H) \setminus \{0\}$.
Fix $v,u\in \vN$ such that $\gcd{(v, |S'|) }=\gcd{(v, q^d-1) }=1$ and $uv\equiv 1 \mod |S'| $.
Set 
\[\begin{aligned}
\widetilde{\psi}\colon S'&\longrightarrow S'\\
[s]&\longmapsto [s^v].
\end{aligned}\]

\begin{remark}$\,$
\begin{itemize}
\item{$\widetilde{\psi}^{-1}\colon [z]\longmapsto [z]^u$}
\item{$k[x]/H\cong \vF_{q^d}$ is again a finite field}
\end{itemize}
\end{remark}
Let now $L\in \vN$ such that there exist $L$ distinct irreducible monic polynomials
{$p_1,\dots, p_L\in\vF_q[x]$} with the property
\begin{equation}\label{degreeproperty}
\sum_{i=1}^L \deg{p_i}<d.
\end{equation}
Notice that in the present description of the protocol there are several different strategies to choose the polynomials; we will analyse the properties of some interesting choices in the following sections.

Again, we have the encryption map:
{\begin{prop}
$(k[x],\cdot)$ is an $L$ cryptable monoid with the map
\begin{equation}\begin{aligned}
\alpha_\sim^\psi: \vZ_2^L &\longrightarrow S'\\
m=(m_1,\dots,m_L)&\longmapsto \left[\prod^L_{i=1}p_i^{vm_i}\right].
\end{aligned}\end{equation}
\end{prop}
\begin{proof} Definition \ref{defcryptable} requires that the map $\alpha_\sim^\psi$ be an injection. Assume
\[\alpha_\sim^\psi(m_1,\dots,m_L)=\alpha_\sim^\psi(n_1,\dots,n_L)\]
\[ \left[\prod^L_{i=1}p_i^{vm_i}\right]= \left[\prod^L_{i=1}p_i^{vn_i}\right].\]
It follows
\[ \left[\prod^L_{i=1}p_i^{vm_i}\right]^u= \left[\prod^L_{i=1}p_i^{vn_i}\right]^u\]
\[ \left[\prod^L_{i=1}p_i^{m_i}\right]= \left[\prod^L_{i=1}p_i^{n_i}\right]\]
where, in the last equation, we can assume no reduction has happened, since property
(\ref{degreeproperty}) holds. Indeed
\begin{equation}\label{canlift2} \prod^L_{i=1}p_i^{m_i}= \prod^L_{i=1}p_i^{n_i}.\end{equation}
Recalling that $k[x]$ is a unique factorization domain we have $m_i=n_i\  \forall i$.
\end{proof}
}

So our cyphered text is given by $c(x)=\alpha_\sim^\psi(m_1,\dots,m_L)$. The explicit decryption for this protocol is simply given by the polynomial division of the decyphered code $(c(x))^u$, that is to say 
\begin{equation}\label{decryption}
m_i=1 \Longleftrightarrow (c(x))^u=0 \mod p_i(x) .
\end{equation}

\begin{remark}
We stress once again the fact that in obtaining equation \eqref{canlift2} we used the canonical lift 
\begin{align*}
\Gamma\colon &S/\sim\,\longrightarrow S \\
& [f(x)]\longmapsto g(x)
\end{align*}
where, for any representative $l(x)\in[f(x)]$, $g(x)$ is the remainder of the division of $l(x)$ by $h(x)$ in $k[x]$, and it is obviously independent of the choice of $l(x)$.
The decryption is effectively performed in $\Im(\alpha)$ and the solution to Problem \ref{problem2} is then given by $(\alpha^{-1}\circ \Gamma)(c(x)^u)$.
\end{remark}

The information rate $\mathcal{I}=L/\deg(h)\log_2(q)$ depends on the choice of the carrier polynomials. We will explain later how to maximise this value.

\begin{remark}
Once the $p_i$'s are fixed the top information rate for this protocol is obtained when we choose $h(x)$ such that
\begin{equation}\label{topinfocondition}
\sum\limits_{i=1}^L\deg{p_i}=\deg{h}-1.
\end{equation}
Indeed the information rate can always be maximised since it is always possible to choose $h(x)$ in $k[x]$ such that \eqref{topinfocondition} is satisfied (cf. Remark \ref{NSKleak}) without allowing for a straightforward reduction to a DLP. This case will be analysed in detail in \ref{efficiency}.
\end{remark}

\subsection{A simple example}
We now give an example in which $k[x]= \vF_2[x]$ and the space of messages has size $2^9$. In order to reach a message size of $9$ bits, we need exactly $9$ keys, that is to say monic irreducible polynomials in $\vF_2[x]$. From finite field theory, we know that there are exactly $q$ monic polinomials of degree 1, and 
\[
\frac{q^d-q}{d}
\]
irreducible monic polynomials of prime degree $d$. So, for $q=2$ we have two polynomials of degree 1, one polynomial of degree 2, two polynomials of degree 3 and six polynomials of degree 5. For the sake of simplicity, even if the example is non optimal as we will explain, let us choose all the irreducible monic polynomials of degree 1,2 and 5, summing up to exactly 9 keys, namely:
\begin{gather}
p_1=x\\
p_2=1 + x\\
p_3=1 + x + x^2\\
p_4=1 + x^2 + x^5\\
p_5=1 + x^3 + x^5\\
p_6=1 + x + x^2 + x^3 + x^5\\
p_7=1 + x + x^2 + x^4 + x^5\\
p_8=1 + x + x^3 + x^4 + x^5\\
p_9=1 + x^2 +x^3 + x^4 + x^5.
\end{gather}

Then, the public key $h(x)$ must be of degree 
\[d=\mathrm{deg}(h(x))=\sum_{i=1}^9\mathrm{deg}(p_i(x)) + 1 = 35\]
and irreducible. For instance we may take
\begin{equation}
h(x)=1+x^2+x^{35}
\end{equation}
and set our protocol onto $\vF_{2^{35}}\cong(\vF_2[x]/H)^*$, whose order is $2^{35}-1$ when $H=\left<h(x)\right>$.
We choose the secret key and the decryption exponent, accordingly, to be $v=3821$ and $u=25169564954$, so that $uv=1\ \mathrm{mod}(2^{35}-1)$.
Then we may publish the 9 carrier keys $p_i^v \mod {(h(x),2)}$:
\begin{align}
p_1^v=&1 + x^2 + x^4 + x^{10} + x^{12} + x^{18} + x^{22} \\ 
&+ x^{23} + x^{24} + x^{26} + x^{27} + x^{29} + x^{32}\notag\\
p_2^v=&x + x^3 + x^5 + x^6 + x^7 + x^{10} + x^{12}+ x^{13}\\
&  + x^{17} + x^{20} + x^{21} + x^{22} + x^{24} + x^{28} + x^{30} + x^{32}\notag\\
p_3^v=&x + x^4 + x^5 + x^7 + x^{13} + x^{20}+ x^{22} \\ 
& + x^{28} + x^{29} + x^{30} + x^{31} + x^{32} + x^{33} + x^{34}\notag\\
p_4^v=&1 + x^2 + x^3 + x^4 + x^{11} + x^{14} + x^{15} + x^{17} + x^{18}\\
& + x^{19} + x^{20} + x^{21} + x^{24} + x^{28} + x^{30} + x^{34}\notag\\
p_5^v=&1 + x + x^2 + x^3 + x^4 + x^7 + x^8 + x^9 + x^{10} + x^{11} + x^{15}\\
& + x^{18} + x^{20} + x^{21} + x^{22} + x^{24} + x^{26} + x^{29} + x^{32} + x^{33}\notag\\
p_6^v=&1 + x + x^2 + x^4 + x^7 + x^{12} + x^{13} + x^{15} + x^{16} +\\
& x^{18} + x^{21} + x^{22} + x^{23} + x^{24} + x^{30} + x^{34}\notag\\
p_7^v=&1 + x^4 + x^8 + x^9 + x^{10} + x^{15} + x^{19} + x^{28} + x^{30} + x^{32} + x^{33}\\
p_8^v=&x + x^3 + x^4 + x^5 + x^8 + x^{10} + x^{12} + x^{13} + x^{15} + x^{16} \\
&+ x^{17} + x^{25} + x^{26} + x^{27} + x^{28} + x^{30}\notag\\
p_9^v=&x + x^4 + x^6 + x^7 + x^{10} + x^{11} + x^{12} + x^{13} + x^{14} + x^{15} + x^{16}\\
& + x^{17} + x^{18} + x^{20} + x^{23} + x^{24} + x^{30} + x^{31} + x^{32} + x^{33}.\notag
\end{align}

Suppose we want to send the message $m=111000111\in{\vZ_2^9}$, we encode it into
\begin{align}
c=&\prod\limits_{i=1}^9p_i^{vm_i} \mod {(h(x),2)}=\notag\\
=&\ x^2 + x^3 + x^6 + x^{10} + x^{15} + x^{16} + x^{17} + x^{18}\notag\\
+&\ x^{20} + x^{21} + x^{23} + x^{26} + x^{27} + x^{30} + x^{31} + x^{33} + x^{34}.
\end{align}

Once the message has been received, it is sufficient to take the {$u$-th} power, and the result is as follows:
\begin{align}
c^u=&\prod\limits_{i=1}^9p_i^{vum_i}  \mod {(h(x),2)}=\prod\limits_{i=1}^9p_i^{m_i} =\notag\\
=&\ x + x^3 + x^4 + x^6 + x^{11} + x^{12} + x^{14} + x^{15} + x^{16} + x^{19}
\end{align}
whose factorization yields:
\begin{equation}\begin{aligned}
\mathsf{Factor}_{2}(c^u)=&x (1 + x) (1 + x + x^2) (1 + x + x^2 + x^3 + x^5)\\
& (1 + x + x^3 + x^4 + x^5) (1 + x^2 + x^3 + x^4 + x^5).
\end{aligned}\end{equation}
We used the factorization algorithm in this simple example because we are working with small messages. The decryption algorithm presented in \eqref{decryption} is to be considered preferential.

The information rate associated to this encryption protocol is 
\begin{equation}
\mathcal{I}=\frac{L}{\deg(h)}=\frac{9}{35}\cong 25,7\%
\end{equation}
with the size of the space of messages being $2^9$.
\begin{remark}
A similar example is presented in \cite{NSK}, with $2^8$ messages. In the cited example the information rate is slightly higher than ours, yet comparable, but the space of messages is smaller.

If we wanted to match the size of space of messages it would be sufficient to remove one polynomial of degree 5, obtaining an information rate of $\mathcal{I}=8/30\sim 26,7\%$.

Remarkably enough, as in the NSK-protocol there is apparently no key leakage, our protocol preserves the security of the carrier keys. As a matter of fact, factoring the cyphertext $c$, one gets no information whatsoever on the cleartext, as it can be seen in the given example:
\begin{multline}
\mathsf{Factor}_2(c)=x^2 \left(x^4+x^3+1\right) (x^{28}+x^{25}+x^{24}+x^{23}+x^{22}+x^{21}+x^{20}+x^{18}+\\
x^{17}+x^{15}+x^{14}+x^{12}+x^{11}+x^{10}+x^8+x^6+x^5+x^4+x^3+x+1)
\end{multline} 
\end{remark}

{
\begin{remark}
More generally, let $g(x)$ be the public modulus and 
\[p_1^{vm_1} p_2^{vm_2}\cdots p_L^{vm_L}\equiv c(x) \mod g(x)\] a cyphetext. Observe that over $\vF_{q}[x]$ we have   \[P(x)=p_1^{vm_1} p_2^{vm_2}\dots p_L^{vm_L}=t(x)g(x)+c(x)\] for some $t(x)\in\vF_{q}[x]$. Now notice that infering on the factorization of $P(x)$ from the data of $c(x)$ in terms of the factor basis \[\{p_1^{vm_1},\dots ,p_L^{vm_L}\}\] is the difficult problem on which the protocol relies, since the factorization of polynomials behaves badly with respect to reductions modulo irreducible polynomials. 
As a matter of fact, we base the security of our protocol on the randomness of the factorization of elements in the image of the map
\[\Gamma_{g,c}: \vF_q[x] \longrightarrow \vF_{q}[x]\]
\[\Gamma_{g,c}(t(x))=t(x)g(x)+c(x)\]
In general, the usual security one expects using prime numbers as carriers (NSK) can be extended to monic irreducible polynomials.

\end{remark}

 }

As we already { pointed out}, we are using here a non-optimal setting for our example, in that we skipped the polynomials of degree 3 and 4, and used all those of degree 5 instead. If we decided to optimize the information rate, we could take the two polynomials of degree 1, the single polynomial of degree 2, two of degree 3 and three of degree 4, for an overall encoding power of $2^8$ messages. Notice that the space of messages is again equal to the example given in \cite{NSK}. 

Choosing polynomials of degree 3 and 4 instead of 5 allows us to reduce the degree of $h(x)$, that is to say the number of bits that are needed to encrypt a message. So, if we compute the information rate in this case we obtain a much better result:
\begin{equation}
\mathcal{I}=\frac{\log_2{m}}{\log_2{c}}=\frac{8}{23}\cong 34,78\%
\end{equation}
which is slightly higher than the information rate presented in \cite{NSK} for the same message size.

The procedure works exactly the same when we change the ground field from $p=2$ to $p=3$. This time we may choose three polynomials of degree 1, three of degree 2 and two of degree 3, all monic and irreducible, allowing us to reduce the overall degree of $h(x)$ to $\mathrm{deg}(h(x))= 16$. In this case, for the same message size, we get an information rate of
\begin{equation}
\mathcal{I}=\frac{8}{16\log_2 3}\cong 31,55\%
\end{equation}
which is not better than the information rate in \cite{NSK}, for a space of messages of the same size, yet comparable.

\subsection{Flexibility of the protocol}
We have already pointed out in the previous sections that the important condition \eqref{degreeproperty} can be fulfilled in several different ways according to the strategy we use in choosing the carrier polynomials $p_i$'s.
In what follows we will present a strategy that optimises the information rate and one that, to our analysys, improves security.

{We will give a detailed analysis of the asymptotics of the information rate of our protocol and of NSK, showing that they have the same behaviour. In what follows our finite field $k$ will be $\vF_q$ for some prime power $q$.
}

\subsubsection{Optimization of the information rate}\label{efficiency}
The optimization of the information rate is ensured by the following:
\begin{prop}
There exists a strategy that maximises the information rate $\mathcal{I}$ for any choice of $q$ and $L$. Moreover, in this strategy the information rate is determined by the closed formula
\begin{equation}\label{inforate}
\mathcal{I}(q,N)=\frac{\sum\limits_{n=1}^N\frac{1}{n}\sum\limits_{k\mid n}\mu\left(\frac{n}{k}\right)q^{{k}}}{\left( \sum\limits_{n=1}^N\sum\limits_{k\mid n}\mu\left(\frac{n}{k}\right)q^{{k}}+1\right)\log_2q}
\end{equation}
where $\mu(x)$ is the M\"obius function.
\end{prop}
\begin{proof}
We defined the information rate to be $\mathcal{I}=L/(\deg{h}\log_2q)$ and we know that the degree of $h$ depends on the particular choice of carrier polynomials. The strategy we will consider is simply given by choosing \emph{all} irreducible polynomials of all degrees  up to a given degree $N$. Denote the number of degree-$n$ irreducible polynomials in $\vF_q[x]$ by $D^q_n$, we have the formula
\begin{equation*}
D_n^q=\frac{1}{n}\sum_{k\mid n}\mu\left(\frac{n}{k}\right)q^k
\end{equation*}
where $\mu(x)$ is the M\"obius function. The overall number of chosen polynomials, that is the number of bits that the plain text is composed by, as well as the sum of the degrees of the $p_i$'s are given by a closed formula, namely:
\begin{align}
L&=\sum\limits_{n=1}^N D_n^q= \sum\limits_{n=1}^N\sum\limits_{k\mid n}\mu\left(\frac{n}{k}\right)\frac{q^{{k}}}{n}\\\label{degh}
\deg(h(x))&=\sum_{n=1}^N nD_n^q+1=\sum\limits_{n=1}^N\sum\limits_{k\mid n}\mu\left(\frac{n}{k}\right)q^{{k}}+1
\end{align}
for some maximal degree $N$ (which is dependent on $L$ if we consider $L$ to be the fundamental parameter). Then, the information rate $\mathcal{I}$ as a function of the prime power $q$ and (implicitly) the parameter $L$ has the desired closed expression.

It is easy to gather that such a choice of the polynomials guarantees maximal information rate, in that we are lowering as much as possible the degree of $h(x)$ and as a result the number of bits of the encrypted message.
\end{proof}

\begin{remark}
The obvious disadvantage of the strategy above is that one can always assume that the bare carrier polynomials are known, for we take all of them progressively up to degree N. As a matter of fact, the strategy above gives us a clear upper bound for the information rate, for all different combinations of $L$ and $q$. Notice, however, by comparison with the tables of \cite{NSK}, that this is the same strategy adopted by Naccache and Stern, where the chosen prime $p$ has the same size of $\mathsf{NextPrime}(\prod p_i)$.
\end{remark}

Within this strategy it is important to notice that all the variations proposed in \cite[Section 2.3]{NSK} are importable in the present context. For example, it is possible to express the message $m$ in a basis different from 2, and this would lead to some modification to the suitable degrees for our carriers. Moreover, it is possible to restrict the space of messages to constant-weight strings. This last choice increases the information rate since it allows to lower the degree of $h(x)$. In fact, if $w$ is the constant weight, the bound on the degree of $h$ is:
\[\deg{h}>wN\]
where N is the highest degree of the chosen carriers.

Apart from these extensions, the standard NSK protocol is summarized in the table presented in \cite[Section 2.2]{NSK}, where the information rate for 512, 1024 and 2048 bits-sized $p$'s is given. The strategy we have just outlined to reach the maximal information rate, allows us to obtain the exact values presented in \cite{NSK} matching the degree of our polynomial $h$ with the size of their prime $p$ and $L$ with the size $\mathcal{M}$ of the message. So we are able to obtain the same information rate.

The matching procedure works as follows: compute the degree of $h$ obtained by choosing all polynomials up to a given degree, say 9 to obtain $\deg{h}=977$. Then, top it to the next block, in this case 1024 bits, choosing \emph{some} polynomials of one degree higher, in this case 11. This leads to an increase in the number $L$ of carrier polynomials from 127 to 131, and the information rate is then given by the ratio $L/\deg{h}$.

\begin{table}
\centering
\begin{tabular}{c|c|c}
$L$ (bits) & $\deg{h}$ (bits) &  $\mathcal{I}$\\\hline
131 & 1024 & 12,8 \%\\
233 & 2048 & 11,4\% \\
418 & 4096 & 11,2\%
\end{tabular}\caption{{Information rate matching with \cite[Section 2.2]{NSK}}}\label{tab1}
\end{table}

{ In Table \ref{tab1} we show how to match the examples presented in \cite{NSK}, and the last row is obtained by extending their calculations to 4096 bits. If we go further and compute the relevant figures in the case of 8192 bits we find almost perfect agreement also in this case (cf. Table \ref{tab2}). It will be clear in what follows why this happens. }
\begin{table}
\centering
\begin{tabular}{c|c|c|c}
$L$ (bits)& $\mathcal{M}$ (bits) & Size of $p$ \& $\deg{h}$ (bits)  &  $\mathcal{I}$\\\hline
759& 758 & 8192  & 11,4\%
\end{tabular}\caption{{Extension to next block and matching of the information rate}}\label{tab2}
\end{table}

{
\subsubsection{Asymptotics comparison with previous works}\label{compare}
We will prove in this section that our protocol has the same asymptotic information rate of \cite{NSK}. A naive explanation of this fact is given by arguing that the number of primes below a certain number of bits has the same behaviour as the number of irreducible polynomials in $\vF_q[x]$ below a certain degree.

Let us fix the notation 
\[
a_N\sim \,b_N\ \ \Longleftrightarrow \lim_{N\rightarrow\infty} \frac{a_N}{b_N}=1.
\]
}
{
We will make use of the following
\begin{lemma}\label{lemlem}
\begin{equation}\label{lemma}
\sum\limits_{n=1}^N D_n^q \sim  \frac{q}{q-1} D^q_N
\end{equation}
\end{lemma}
\begin{proof}
First recall that \cite[Theorem 2.2]{Poly}
\[
D_n^q \sim \frac{q^n}{n}
\]
and therefore the sums behave asymptotically as
\[
\sum\limits_{n=1}^N D_n^q \sim \sum\limits_{n=1}^N \frac{q^n}{n}.
\]
Then we have \eqref{lemma} if and only if
\begin{equation}
\lim_{N\rightarrow \infty} \frac{\sum\limits_{n=1}^N \frac{q^n}{n}}{\frac{q^N}{N}} = \frac{q}{q-1}.
\end{equation}

Now, denote by $S_N\mathrm{:=}\sum\limits_{n=1}^N \frac{N}{n}q^{n-N}$ and observe that it might be expressed in terms of the recursive sequence
\begin{equation}\label{recursion}
S_{N+1} = \frac{1}{q}\frac{N+1}{N} S_N +1.
\end{equation}
for the initial value $S_1=1$. Consider $S_-=\liminf_{N\rightarrow\infty}S_N$ and $S_+=\limsup_{N\rightarrow\infty}S_N$. Passing to the $\limsup$ and $\liminf$ in \eqref{recursion} we get the same equation for $S_\pm$:
\[
S_\pm=\frac{S_\pm}{q} +1
\]
provided that they are both finite. Assuming that they are, we conclude that 
\begin{equation}
\lim_{N\rightarrow\infty}S_N=S_\pm=\frac{q}{q-1}
\end{equation}

This assumption is legitimate since $S_N\geq 0$ for all $N\in \mathbb{N}$, thus $S_-\geq 0$, and for $S_+$ we observe that 
\begin{itemize}
\item When $x\in\vR^+$ we have that $\frac{q^x}{x}$ is increasing for $x\geq \frac{1}{\log{q}}\geq 2$, since $q\geq 2$, and in particular this is true for $x\in\vN^*$;
\item $\limsup_{N\rightarrow \infty} \frac{N}{q^N}\sum\limits_{n=1}^N \frac{q^n}{n}=\limsup_{N\rightarrow \infty} \frac{N}{q^N}\sum\limits_{n=2}^N \frac{q^n}{n}$.
\end{itemize}
It follows that 

\[
\limsup_{N\rightarrow \infty}\frac{N}{q^N}\sum\limits_{n=1}^N \frac{q^n}{n} = \limsup_{N\rightarrow \infty} \frac{N}{q^N}\sum\limits_{n=2}^N \frac{q^n}{n} \leq \limsup_{N\rightarrow \infty}\frac{N}{q^N}\int\limits_{2}^{N+1}\frac{q^x}{x}dx
\]
where the last inequality comes from the fact that $\sum\limits_{n=2}^N \frac{q^n}{n}$ are the lower sums of $\int\limits_{2}^{N+1}\frac{q^x}{x}dx$, since $\frac{q^x}{x}$ is increasing for $x\geq 2$. Moreover
\[
\lim_{N\rightarrow \infty}\frac{\int\limits_{2}^{N+1}\frac{q^x}{x}dx}{\frac{q^N}{N}} = \lim_{t\rightarrow \infty}\frac{\int\limits_{2}^{t+1}\frac{q^x}{x}dx}{\frac{q^t}{t}}=\lim_{t\rightarrow \infty}\frac{\frac{q^{t+1}}{t+1}}{\frac{q^t}{t}(\log{q} - \frac{1}{t})}=\frac{q}{\log{q}}
\]
where the second equality follows from the De L'H\^{o}pital rule. This proves that
\[
0\leq \liminf_{N\rightarrow \infty} S_N \leq \limsup_{N\rightarrow \infty} S_N \leq \frac{q}{\log{q}}
\]
and yields the claim.
\end{proof}
}

{ We are now ready  to prove

\begin{prop}\label{MSKprop}
\begin{equation}\label{asympt}
\mathcal{I}(q,N)\sim \frac{1}{\log_2{q}}\frac{1}{N}
\end{equation}
\begin{proof}
Observe that $nD_n^q\sim q^n$ and therefore, from \eqref{inforate}
\[
\mathcal{I}(q,N)\sim\frac{\sum\limits_{n=1}^N \frac{q^n}{n}}{\log_2{q}\sum\limits_{n=1}^N q^n}
\]
Now, it is easy to gather that 
\begin{equation}\label{geoseries}
\sum\limits_{n=1}^N q^n \sim \frac{q}{q-1} q^N
\end{equation}
then, plugging the results of \eqref{geoseries} and of Lemma \ref{lemlem} into \eqref{inforate}, we obtain
\begin{equation}
\mathcal{I}(q,N)\sim \frac{1}{\log_2{q}}\frac{\frac{q}{q-1}\frac{q^N}{N}}{\frac{q}{q-1}q^N} =\frac{1}{\log_2{q}} \frac{1}{N}.
\end{equation} 

\end{proof}
\end{prop}
{
We would like to compare this result with the information rate of the NSK protocol. Notice that in order to make a consistent comparison we must understand the role of our parameter $N$ in the NSK.  

Once $q$ is fixed,  bounding the degree of the carrier polynomials by $N$ is the same as bounding
 the number of bits required to represent any of them by the quantity $M=\left\lfloor N \log_2(q)\right\rfloor$. 

The analogous bound for the NSK is then given by bounding the number of bits of the prime carriers by $M$. This is the same as bounding the prime carriers themselves by $2^M\simeq q^N$.
In the following proposition the comparison is made explicit.}
{

\begin{prop}\label{NSKprop}
Let $N$ be the bound on the degree of the carrier polynomials and $M=\left\lfloor N \log_2(q)\right\rfloor$ the analogous bound for the bits of the prime carriers in the NSK.
The information rate for the NSK protocol  is asymptotically given by
\begin{equation}
I_{NSK}\sim \frac{1}{\log_2{q}}\frac{1}{N}.
\end{equation}
\begin{proof}
It is known  \cite[Equation 2]{Erd} that for large $m\in\vN$
\[
\prod\limits_{p<m}p \sim e^{m}.
\]
Let us consider $m=2^M\simeq q^N$, then $\prod_{p<q^N}p \sim \exp{q^N}$. Now, the number of prime numbers up to $q^N$ asymptotically goes, by the prime number theorem, as
\[
\pi(q^N)\sim \frac{q^N}{N\ln{q}}.
\]
In our case this will be the number of carrier prime numbers up to $q^N$. On the other hand $\exp{q^N}$, which is the size of  the prime modulus of \cite{NSK}, has $\left\lfloor q^N\log_2{e}\right\rfloor$ digits, and therefore the information rate is computed as
\begin{equation}
I_{NSK}\sim \frac{\frac{q^N}{N\ln{q}}}{q^N\log_2{e}}=\frac{1}{\log_2{q}}\frac{1}{N}
\end{equation}
\end{proof}
\end{prop}}

{
By comparing Propositions \ref{MSKprop} and \ref{NSKprop} it is now clear that the two information rates have the same behaviour.  This explains that the matching procedure we perform at the end of the previous section will attain the information rate of NSK also in the asymptotic limit. Moreover it justifies the claim on the large-$N$ behaviour of irreducible polynomials with respect to prime numbers. }

}

\subsubsection{Some precautions to avoid subgroup-like attacks}\label{security}
The security of this protocol is strictly related to the size of the degree of $h$ and, as a consequence, to  the range of degrees that the carriers can have. Indeed, when the carriers are chosen within a large set, the attacker will not have chances (in terms of a brute force attack) to find the $p_i$'s to set up a discrete logarithm problem for the pair $(p_i,p_i^s)$ for  any $i$.

As a matter of fact, the knowledge of $h$ will only lead to the following information on the degrees:
\[\deg(h)=\sum_i{\deg(p_i)}+1.\] 
This is not the case when working with integers and primes in $\vZ/p\vZ$, where we can always assume that the prime factors are known
when $p\simeq \prod_i p_i$.

We first sketch a subgroup like attack in the most \emph{unsafe} case.
Let $G$ be an abelian group and $p^v_1,\dots,p^v_L$ be carriers, as in Section \ref{poly}.
Let the order of $p^v_i$ in $G$ be $n_i$ and suppose $\gcd{(n_i,n_j)}=1$ for $i\neq j$.
Let now
\[M_j=n_1\cdots n_{j-1}\cdot n_{j+1}\cdots n_L.\]
It is easy to observe that, for a generic cyphertext $c$, $m_j=1$ if and only if $c^{M_j}\neq 1$. As it is elementary to observe, this leads to decryption in $L$ steps. Moreover, it can also be adapted to work when the condition $\gcd{(n_i,n_j)}=1$ is just partially fulfilled. {In this case, indeed, only partial information on the text can be extracted.}

Consider now the decomposition in cyclic subgroups of the multiplicative group of the finite field $(\vF_{q^d})^*$. In order to avoid subgroup-like attacks on the cyphertext we will require all the $p_i$'s to be generators of the same subgroup of large order. This will lead to certain requirements on $q^d-1$.

The most natural choice to solve this problem is asking that the degree $d$ of the reducing polynomial $h(x)$ be constrained by the following:
\begin{equation}\label{chone}
r:=\frac{q^d-1}{q-1}\ {\mathit{is\ prime}}.
\end{equation}

Now we have to choose the $p_i$'s such that
\begin{equation}\label{chtwo}
p_i(x)^r\neq 1 \mod{h(x)}\quad \forall i\in\{1,\dots,L\}.
\end{equation}
When these conditions are satisfied, all the $p_i$'s are generators for $(\vF_{q^d})^*$.

\subsection{``Chinese remainder'' version}\label{CRT}
In what follows we will present another example of a protocol that fits the general picture, which stems on the well known chinese remainder theorem. To do this, let us introduce a large prime power $q$ and a natural number $L\in \vN$. Consider now the monoid $S=(\vF_q^{L+1})^*$, with the multiplication defined componentwise, and the set $R=\{r_1,\dots,r_{L+1}\}\subseteq\vF_q$.

Let $\alpha_i\in \vF_q\backslash R\quad\forall i\in\{1,\dots,L\}$ and choose two large integers $u,v$ such that $uv=1\mod{(q-1)}$. Compute the following list of vectors $p_i\in(\vF_q^{L+1})^*$ as
\[
(g_i)_j:=(r_j-\alpha_i)
\]
\[
(p_i)_j:=(g_i)^v.
\]
Let 
\[((\vF_q^{L+1})^*, \{p_1,\dots,p_L\})\]
be the public key and
\[(\{g_1,\dots, g_L\}, \{r_1\dots r_L\})\]
be the secret key.
Let
\[F:\vZ_2^L \longrightarrow S\]
\[(m_1,\dots,m_L)\mapsto  \prod_{i=1}^L p_i^{m_i}\]
be the encryption map.
\begin{remark}
Observe that the information rate is
\[\frac{L}{(L+1)\log_2(q)}.\]
\end{remark}
\begin{prop}
$F$ is an injection.
\end{prop}
\begin{proof}
We define a polynomial on $\vF_q[x]$ by
\[
h_R(x):=\prod_{i=1}^{L+1}(x-r_i)
\]
whose set of zeros coincide with $R$.
We will prove the proposition by showing how to compute the inverse over the image of $F$ using $h(x)$, i.e. we will show how to uniquely decrypt any cyphertext $c\in \Im(F)$ using the secret key.
Let 
\[\psi: S \longrightarrow S\]
\[x\mapsto x^v\]
\[G: \vF_q[x]/h_R(x) \stackrel{\mathrm{CRT}}{\longrightarrow} \vF_q^L\]
\[k(x)\mapsto (k(r_1),\dots,k(r_L))\]
and 
\[\Gamma: \vF_q[x]/h_R(x) \longrightarrow \vF_q[x]\]
be the canonical lift.
The decryption map $D$ is given by checking $\Gamma( G^{-1}(\psi^{-1}(x)))$ modulo $g_i(x)=(x-\alpha_i)$: whenever it is zero it means $m_i=1$, where $ \psi^{-1}(x)= x^u$. Observe that the decryption is well defined: the map
\[\alpha_\sim^\psi: \vZ^L_2\longrightarrow \vF_q[x]/(h_R(x))\]
is clearly injective (and then $\alpha_\sim$ is, by Proposition \ref{injalpha}) since the product of all the $g_i(x)$ has degree $L<L+1$.
Observe that $\sim$ is as usual the relation induced by the ideal of $h_R(x)$.
\end{proof}

\section{Outlook and further research}
In the present communication we have given a new setting to produce many examples of knapsack encryption schemes, showing also how a remarkable example such as \cite{NSK}  perfectly fits our framework. We have proposed a next-to-simplest example when the monoid is chosen to be $(k[x],\cdot)$, one realization of which is given by $\vF_q[x]$ reduced by the ideal of an irreducible polynomial of suitable degree. 
{ 

This brand new application of the knapsack idea reproduces the key results presented in \cite{NSK} in terms of information rate, but allows us to improve some important features such as  
\begin{itemize}
\item the information rate is shown to be deterministic by providing an exact formula for it (cf. \cite[Section 2.2]{NSK}).
\item it reduces the computations over $\vF_{q^d}$ with $p\sim q^d$ but $q\ll p$, where $\vF_q$ is a field of small characteristic.
\end{itemize}

A non trivial variation of this scheme has been found, by taking into account a polynomial which splits over the base field and applying the chinese remainder theorem, allowing the computations to be performed over a direct sum of fields.

}
In \cite{NSK} Naccache and Stern conjectured that it might be possible to elliptic curve their scheme, and the new general framework we have presented might be of some help to address this problem.

Moreover, it would be interesting to see how the recent improvements to the NSK protocol presented in \cite{CHEV} may apply to our polynomial instance. This will be matter of further studies.

\section*{Acknowledgements}
The authors would like to thank Patrik K\"uhn, G\'erard Maze, Joachim Rosenthal and Davide Schipani for helpful discussions and suggestions. M.S. acknowledges partial support from SNF grant 200020\_149150/1.

\end{document}